\documentclass[journal]{IEEEtran}
\usepackage{amsmath,amssymb,amsthm,graphicx,cite}
\graphicspath{{figs/}}
\newtheorem{theorem}{Theorem}
\newtheorem{corollary}{Corollary}
\newcommand{\pstar}{p^{\star}}
\newcommand{\qstar}{q^{\star}}
\newcommand{\lsn}{\lambda_{\mathrm{SN}}}

\begin{document}

\title{An Exact Componentwise Voltage-Monotonicity\\
Threshold for Radial Distribution Networks}

\author{Marian~Me\v{s}ter% <-this % stops a space
\thanks{M.~Me\v{s}ter is with the Technical University of Ko\v{s}ice, Slovakia (\texttt{marian.mester@tuke.sk}).}%
\thanks{Manuscript prepared July 2026. The protocol code, the written pre-registrations of the falsification, convergence and Volt-VAR runs, the deviations log, and the complete run archives are permanently deposited at Zenodo (DOI \texttt{10.5281/zenodo.21407024}). The author used Anthropic Claude (a large language model) for language editing and as an auditing aid in verifying the numerical protocol; all analysis, theorems, proofs, numerical results, and conclusions are the author's own and were verified independently.}}

\maketitle

\begin{abstract}
Distribution-network operating tools routinely assume that voltage magnitudes respond \emph{monotonically} to nodal power injections, so that checking the extreme points of an operating envelope certifies its interior. That the voltage may instead \emph{fall} under rising active export is documented---measured on feeders and explained through phasor loci and P--V maxima---but its exact onset has not been placed in closed form. We derive the exact sign of the single-branch voltage sensitivity and show that componentwise monotonicity in active power reverses precisely at $\pstar = g\,v$, while the reactive threshold $\qstar = b_{a}\,v$ sits at short-circuit scale and is physically unreachable; the two differ by the branch $x/r$ ratio. A test protocol whose scenario grid, tolerances and acceptance criteria were fixed before the runs, on two independent networks ($1594$ cells, $26\,000$ ordered injection pairs), corroborated on a second real network and a conductance sweep, places every observed violation at the low-$R/X$ station branch, and a per-cell attribution separates threshold crossings from a second, loss-mediated coupling that reverses sensitivities while every individual branch margin is still positive. Behind a regulated busbar the entire protocol passes, so a network-scale sufficiency condition is stated as an explicit, evidence-supported hypothesis. Because every power-flow model affine in the injections is componentwise monotone by construction, screening built on such models cannot detect this failure. The threshold reads operationally as a per-branch monotonicity margin, giving distribution operators a closed-form flag for where extreme-point screening ceases to be valid.
\end{abstract}

\begin{IEEEkeywords}
Distribution networks, distributed generation, load flow, monotone systems, voltage sensitivity, hosting capacity.
\end{IEEEkeywords}

\section{Introduction}
\IEEEPARstart{M}{any} operational and planning functions assume that voltage magnitudes respond \emph{monotonically} to nodal power injections: interval-based screening, robust solvability certificates \cite{nguyen2019}, and hosting-capacity rules for distributed energy resources (DER) declare an operating envelope feasible whenever its extreme points are \cite{miu2000}. The value of the assumption is computational: if voltage is monotone in each injection over an operating box, the whole box is certified by its corners, and the cost of checking a feasibility or hosting question collapses from a dense sweep to a handful of evaluations.

The rigorous monotonicity results certify strictly less than this practice assumes. Dvijotham \emph{et al.} \cite{dvijotham2017} prove monotonicity in the \emph{rotated} injections $\tilde q = q + \kappa p$ at fixed $\tilde p = p - \kappa q$---a cone-order result, not a componentwise one. Miu and Chiang \cite{miu2000} prove monotonicity of the feasible radial three-phase solution \emph{along a single load-change ray}, and already exhibit a counterintuitive example in which a load decrease lowers a phase voltage; their proof turns on a scalar $g(\lambda)=s\lambda^2+t\lambda$ with $s\propto(r^2+x^2)$, the same impedance combination that governs the threshold below. Wang, Chiang and Guo \cite{wcg2017} give, for unbalanced networks, \emph{sufficient conditions under which} the phase-voltage magnitude is monotone in local active and reactive injections---angular inequalities that certify a region in which monotonicity \emph{holds}, and which their own counterexamples show can be violated, but without locating the onset of that violation. Linear models---LinDistFlow \cite{baran1989} and the error-bounded model of Bolognani and Zampieri \cite{bolognani2016}---are componentwise monotone \emph{by construction}: both are affine in $(p,q)$ with $\partial v/\partial p \ge 0$ entrywise on a radial tree, so there monotonicity is a property of the approximation, not of the network.

A second, physical line has documented the reversal itself. Iioka \emph{et al.} \cite{iioka2019} \emph{measure} voltage reduction under reverse power flow on 6.6\,kV feeders and rationalise it with a two-node P--V analysis, locating a maximum-voltage point along the injection trajectory. Matsumura \emph{et al.} \cite{matsumura2018} describe the same rise-then-fall through phasor diagrams. Das \emph{et al.} \cite{das2022} find, for an inverter constrained to $P^2+Q^2=S^2$, that the terminal voltage is non-monotone in the $P/Q$ ratio, with a stationary point at $P/Q=r/x$. A power-transfer line reaches the same two-bus relation from the loadability side: Vournas \cite{vournas2015} rotates the injections by the loss angle and gives the reactive support for maximum power transfer at a voltage-controlled bus---exactly the reactive threshold $\qstar$ below---while Deakin \emph{et al.} \cite{deakin2018} take a marginal loss-induced maximum-power-transfer limit along the upper voltage limit.

What remains open is narrow but consequential. None of these results gives the \emph{exact componentwise} sign boundary of $\partial|V|/\partial p$ at \emph{fixed} reactive injection, in closed form, together with its reactive counterpart and the asymmetry between them. The measured and P--V accounts \cite{iioka2019,matsumura2018,das2022} differentiate along an injection \emph{trajectory} (a fixed power factor, or the apparent-power circle), which is the total derivative $\mathrm{d}v/\mathrm{d}p = \partial v/\partial p + \alpha\,\partial v/\partial q$, not the partial that an interval screen over an axis-aligned box actually needs; the sufficient-condition results \cite{miu2000,wcg2017} certify \emph{where monotonicity holds} rather than locating \emph{where it fails}, and do so through angular inequalities rather than a closed scalar. This paper settles the partial, componentwise question exactly.

The relation to Wang, Chiang and Guo \cite{wcg2017} is worth stating precisely, as it is the closest prior work. They prove \emph{sufficiency}: a set of angular inequalities under which monotonicity holds, valid for the general unbalanced three-phase network. Theorem~\ref{thm:sign} proves, for the balanced single branch, an \emph{exact sign}: componentwise monotonicity holds if and only if $p < g\,v_1$. The two are complementary rather than nested. Their result is more general in setting (unbalanced, multi-branch) but certifies only a subset of the true monotone region and is silent on where monotonicity fails; ours is narrower in setting (balanced, per branch) but gives the exact boundary---necessary and sufficient---together with the reactive counterpart $\qstar = b_a v_1$ and the $x/r$ asymmetry between them, none of which a sufficient-only condition can provide. In particular, on the balanced branch the sign of $\partial v_1/\partial p$ equals $\operatorname{sign}(g\,v_1 - p)$ identically (verified at $16\,000$ operating points, zero sign mismatches), so no sufficient condition can be sharper there. This paper is therefore not a special case of \cite{wcg2017}, nor a generalization of it; it settles a question their framework poses but does not answer.

A word on scope. The analysis is balanced and per branch, whereas distribution networks are unbalanced and the sufficient-condition results \cite{miu2000,wcg2017} are three-phase. This is deliberate, not a limitation of convenience: the reversal is an \emph{impedance} phenomenon, not a \emph{phase} one. The threshold $\pstar = g\,v_1 = (r/|z|^2)\,v_1$ depends only on the series $r,x$ of the branch and the local magnitude; it carries no term in phase imbalance or load allocation, and numerically the onset tracks $g\,v$ to within $0.00\,\%$ across the reactive loading $q$ (which shifts $|V_1|$ but not the mechanism). In an unbalanced network the same reversal therefore appears per phase on the same low-$\kappa$ branch, driven by that phase's series impedance; the unbalance of \cite{miu2000,wcg2017} adds phase coupling on top of a mechanism that the balanced branch already isolates in closed form. Treating the balanced branch first is thus the natural way to expose the sign structure without the phase entanglement obscuring it. The contribution is threefold. (i) An exact sign formula for the single-branch voltage sensitivity yields the closed-form thresholds $\pstar = g\,v$ and $\qstar = b_{a}\,v$, with the $x/r$ asymmetry that makes the reactive one unreachable, and a phasor reading of the mechanism (Sections~\ref{sec:thm}--\ref{sec:physics}). (ii) Evidence from a protocol fixed before the runs ($1594$ cells, $26\,000$ ordered pairs), with a four-step check and a per-cell attribution, locates every violation at the low-$R/X$ station branch and separates the branch threshold from a loss-mediated network coupling (Sections~\ref{sec:evidence}--\ref{sec:limits}). (iii) A reading of the threshold as an operational monotonicity margin, delimiting where extreme-point screening remains valid and where full nonlinear evaluation is required (Section~\ref{sec:implications}).

\section{An Exact Sign Theorem}\label{sec:thm}
Consider a slack bus $0$ with $V_0 = |V_0|\angle 0$ and $v_0 := |V_0|^2$, connected through a series impedance $z = r + \mathrm{j}x$ ($r \ge 0$, $x>0$, $|z|^2 = r^2+x^2$) to bus $1$ with net complex injection $s = p + \mathrm{j}q$ (load $p<0$, generation/export $p>0$). Let $v_1 := |V_1|^2$. With $I_1 = (V_1 - V_0)/z$ and $s = V_1 \overline{I_1}$, a direct computation gives $(p+\mathrm{j}q)\bar z = v_1 - V_1\overline{V_0}$; writing $a := rp + xq$ and $b := qr - px$ and taking squared magnitudes yields the exact relation
\begin{equation}\label{eq:R}
v_1 v_0 \;=\; (v_1 - a)^2 + b^2 ,
\end{equation}
a quadratic in $v_1$ with discriminant $\Delta := (v_0 + 2a)^2 - 4(a^2+b^2)$ and $a^2 + b^2 = |s|^2 |z|^2$. Equation \eqref{eq:R} is the classical two-bus relation \cite{vournas2015,deakin2018,iioka2019}; in the rotated coordinates of \cite[eqs.~(7)--(9)]{vournas2015} the locus $\Delta = 0$ is the loadability parabola. The high-voltage (HV, practical) solution is the larger root, on which $2v_1 - (v_0+2a) = \sqrt{\Delta} > 0$.

\begin{theorem}[Sign of the voltage sensitivities]\label{thm:sign}
On the HV solution of \eqref{eq:R}, for any $v_0 > 0$,
\begin{align}
\frac{\partial v_1}{\partial p} &= \frac{2\bigl(r\,v_1 - p\,|z|^2\bigr)}{\sqrt{\Delta}}
 \;=\; \frac{2|z|^2}{\sqrt{\Delta}}\bigl(g\,v_1 - p\bigr), \label{eq:dvdp}\\
\frac{\partial v_1}{\partial q} &= \frac{2\bigl(x\,v_1 - q\,|z|^2\bigr)}{\sqrt{\Delta}}
 \;=\; \frac{2|z|^2}{\sqrt{\Delta}}\bigl(b_{a}\,v_1 - q\bigr), \label{eq:dvdq}
\end{align}
where $g := r/|z|^2$ and $b_{a} := x/|z|^2$ denote the branch conductance and susceptance. In particular,
$\operatorname{sign}(\partial v_1/\partial p) = \operatorname{sign}(g\,v_1 - p)$ and
$\operatorname{sign}(\partial v_1/\partial q) = \operatorname{sign}(b_{a}\,v_1 - q)$. The slack magnitude $v_0$ enters only through $\Delta$; the thresholds are independent of it.
\end{theorem}

\begin{proof}
Differentiate $F := v_1^2 - v_1(v_0+2a) + a^2 + b^2 = 0$ implicitly; on the HV branch $\partial F/\partial v_1 = \sqrt{\Delta} > 0$. With $\partial a/\partial p = r$, $\partial b/\partial p = -x$, $\partial v_1/\partial p = 2[r(v_1-a) + bx]/\sqrt{\Delta}$, whose numerator collapses to $r v_1 - p(r^2+x^2)$---independent of $v_0$. The reactive case is analogous ($\partial a/\partial q = x$, $\partial b/\partial q = r$).
\end{proof}

\begin{corollary}[Monotonicity threshold and the $p/q$ asymmetry]\label{cor:thresh}
Componentwise monotonicity reverses exactly at
\begin{equation}
\pstar = g\,v_1 \quad\text{(active axis)}, \qquad \qstar = b_{a}\,v_1 \quad\text{(reactive axis)}.
\end{equation}
Consequently: (i) $\qstar/\pstar = b_a/g = x/r = 1/\kappa$, so on a low-$\kappa$ branch the reactive threshold stands a factor $1/\kappa$ higher than the active one; $\qstar$ coincides with the reactive support required for maximum power transfer at a voltage-controlled bus \cite[eq.~(14)]{vournas2015}---a short-circuit-scale quantity---so reactive-power monotonicity is unconditionally satisfied at feasible injections. The reactive threshold $\qstar$ therefore reproduces \cite[eq.~(14)]{vournas2015} and is not claimed as new; the contribution on the reactive axis is only its role in the $x/r$ asymmetry. (ii) $\pstar = \kappa v_1/(x(1+\kappa^2))$ collapses as $\kappa \to 0$, so low-$R/X$ branches (transformers, low-loss cables) are the generic locus of monotonicity loss---in the limit $r=0$, $\pstar = 0$ and \emph{any} net active export lowers the voltage.
\end{corollary}

The formulas were checked against central finite differences at $3000$ random points ($\kappa \in [10^{-3},4]$; max.\ deviation $7\times10^{-9}$). Fig.~\ref{fig:signmap} shows the sign map for the CIGRE MV transformer branch.

\begin{figure}[t]
\centering
\includegraphics[width=0.62\columnwidth]{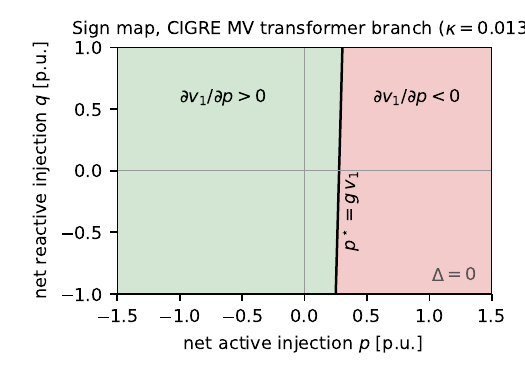}
\caption{Sign of $\partial v_1/\partial p$ in the $(p,q)$ plane for the CIGRE MV transformer branch ($\kappa = 0.013$). The solid curve is the exact threshold $p = g\,v_1(p,q)$; the dashed curve is the existence boundary $\Delta = 0$. Any operating point to the right of the threshold violates componentwise $p$-monotonicity.}
\label{fig:signmap}
\end{figure}

\section{The Mechanism: A Phasor Balance}\label{sec:physics}
The threshold is not a loss-versus-lift competition in the naive sense, and it is worth stating the mechanism precisely because a common shorthand---``the reversal is where the $I^2R$ loss overtakes the resistive rise $rp$''---is quantitatively wrong: the threshold carries $r^2+x^2$, not $r^2$. The correct reading is a phasor balance. Fixing $q$ and writing the drop $V_1 - V_0$ along and across $V_1$, the receiving magnitude obeys $v_0 = (v_1 - a)^2/v_1 + b^2/v_1$ from \eqref{eq:R}; with $a = rp+xq$ and $b = qr - px$, an increment in active export raises the \emph{longitudinal} part of the drop through $rp/|V_1|$---the resistive lift that ordinarily makes injection raise voltage---but also grows the \emph{transverse} part through $xp/|V_1|$, which rotates $V_1$ away from $V_0$ (Fig.~\ref{fig:phasor}). The slack magnitude $|V_0|$ is fixed, so it acts as a budget: the larger the transverse component, the more of that budget it consumes quadratically, and the less remains for the longitudinal projection that sets $|V_1|$. At $p = g\,v_1$ the marginal longitudinal gain and the marginal transverse cost balance exactly; past it, the projection $|V_1|$ turns down even though the injection---and the loss---keep rising. This is why the threshold involves the full $|z|^2 = r^2+x^2$: both channels are reactance-weighted. It is also the same $s\propto(r^2+x^2)$ that controls the ray curvature in \cite[eq.~(5.6)]{miu2000}, seen here as an exact componentwise boundary rather than a bound along a ray.

\begin{figure}[t]
\centering
\includegraphics[width=0.82\columnwidth]{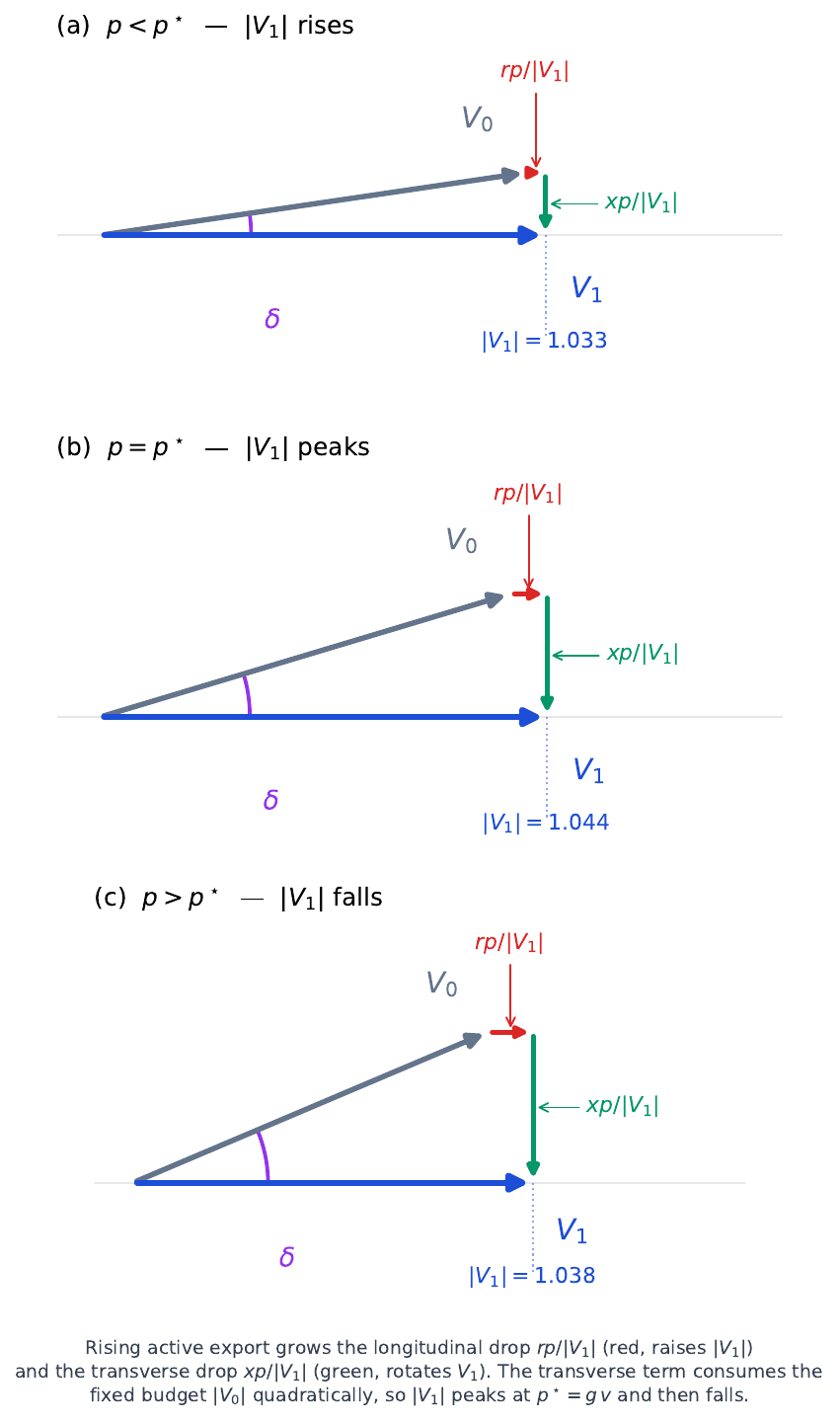}
\caption{Phasor balance on a single branch at fixed $q$, drawn for an illustrative $\kappa$ so the transverse rotation is visible (the mechanism is qualitatively identical for the CIGRE $\kappa=0.013$ branch, where the rotation is small but the same sign change occurs). Rising active export grows both the longitudinal drop $rp/|V_1|$ (red) and the transverse drop $xp/|V_1|$ (green). The transverse term consumes the fixed slack budget $|V_0|$ quadratically, so $|V_1|$ rises for $p<\pstar$, peaks at $p=\pstar=g\,v$, and falls for $p>\pstar$. Inset: at the real CIGRE $\kappa=0.013$ branch the transverse rotation is small, but the sign change is identical.}
\label{fig:phasor}
\end{figure}

Two clarifications follow from the phasor picture and separate this result from the P--V accounts \cite{iioka2019,matsumura2018,das2022}. First, the object here is the \emph{partial} derivative $\partial v_1/\partial p$ at fixed $q$, whereas a measured or plotted P--V curve traces a \emph{trajectory}---a fixed power factor $q=\alpha p$, or the inverter's apparent-power circle $p^2+q^2=S^2$. Along such a trajectory the total derivative is
\begin{equation}\label{eq:total}
\frac{\mathrm{d}v_1}{\mathrm{d}p} = \frac{\partial v_1}{\partial p} + \alpha\,\frac{\partial v_1}{\partial q},
\end{equation}
so a trajectory maximum ($\mathrm{d}v_1/\mathrm{d}p=0$, e.g.\ the $P/Q=r/x$ point of \cite[eqs.~(5)--(6)]{das2022}) mixes both partials and depends on the chosen $\alpha$; the componentwise threshold $\pstar=g\,v$ does not. An interval screen over an axis-aligned $(p,q)$ box asks exactly the componentwise question, and \eqref{eq:dvdp} answers it. Second, the reactive threshold is not a symmetric statement: because $b_a/g = 1/\kappa$, $\qstar$ sits a factor $1/\kappa$ above $\pstar$, at short-circuit scale, and is never reached at feasible injections---so the failure is an active-axis phenomenon specifically.

\section{Network-Scale Evidence}\label{sec:evidence}
The theorem predicts where componentwise monotonicity should fail: on paths crossing branches whose net active flow exceeds $g_e v$. We tested this on the IEEE 33-bus feeder \cite{baran1989} and the CIGRE MV benchmark \cite{cigre2014} in radial operation, both modelled as series $r+\mathrm{j}x$ trees with shunts neglected; the solver was cross-validated against an independent implementation of \emph{that same model} \cite{pandapower2018} to $10^{-10}$~p.u. Restoring cable capacitance and the transformer magnetising branch leaves the picture unchanged---the same DER tier, the same violating row, worst entry $-0.019$ against $-0.019$---so the neglect is not load-bearing. The protocol---scenario grid, tolerances, violation criteria, and the acceptance rules---was fixed in code before the runs and archived with every run; any modification is recorded in a deviations log, and run outputs are timestamped and never overwritten. It sweeps $770$ cells per network (load scaling $\lambda$; DER scaling $\gamma$ to $150\,\%$ of feeder peak; DER placement and power factor; heterogeneous vs.\ homogenized $R/X$), refines the last decade toward the loadability limit ($0.90$--$0.99\,\lsn$, $152$ cells), and evaluates $26\,000$ ordered injection pairs.

A \emph{candidate} is a sensitivity entry below $-\tau$, with $\tau = 10^{-8}$ set far below any magnitude of interest so that screening over-collects rather than under-collects. Discrimination is by a four-step check, applied to every candidate: (i) solver tolerance tightened to $10^{-12}$ and (ii) step refined to $10^{-7}$, lowering the floor to $\approx10^{-8}$; (iii) the sign recomputed with an \emph{independent solver} (Newton--Raphson \cite{pandapower2018}) on the same model and HV branch; (iv) the sensitivity reconstructed in closed form from \eqref{eq:chain}. All $451$ candidates survived all four: none was reclassified, the independent solver confirmed the sign on every one ($99.96\,\%$ entrywise agreement; worst entry $-1.115$ by sweep, $-1.116$ by Newton--Raphson), and the reconstruction matched to $3\times10^{-8}$. The sequence rejected nothing because there was nothing to reject: the weakest candidate is $-1.6\times10^{-4}$, four orders above $\tau$ and above the finite-difference discrepancy at the screening settings. A LinDistFlow cross-check is \emph{inadmissible} here: it is componentwise monotone by construction and can never corroborate a negative sensitivity.

On the IEEE 33-bus feeder ($\kappa \in [0.27,3.0]$, thresholds $\pstar \sim 10^2$~MW---unreachable) no violation occurs: all $770$ grid and $152$ near-boundary cells converge and pass, worst entry $+4.8\times10^{-3}$ up to $0.99\,\lsn$. On the CIGRE MV network, the tree contains the $110/20$~kV transformer branches with $\kappa = 0.013$ and $\pstar = g v \approx 2.8$~MW (at $10$~MVA base). Of the $770$ grid cells, $214$ admit \emph{no HV power-flow solution located by four independent methods} (a flat start, a $5\times10^{4}$-iteration run, a $200$-step homotopy, and a multi-start Newton solve all fail); we cannot exclude a solution beyond the reach of all four, but its absence is consistent across every method. This region is independent of $\lambda$ ($29$--$33$ cells per load tier) and governed by export and lagging power factor ($213$ of $214$ at $\gamma \ge 0.75$; none at $0.95$ leading): an export-side transfer limit across the low-$\kappa$ station branch, not a load-side loadability limit. Monotonicity is undefined there; of the remaining $556$ testable cells, $390$ exhibit structural violations (worst $-1.115$), and of the $152$ near-boundary cells $116$ converge and $61$ violate (worst $-1.918$). A per-cell attribution against exact branch flows splits the $390$ into $135$ cells in which the station-branch flow exceeds $\pstar$---onset at the lowest DER tier at which this is possible ($\gamma = 0.25$, $\approx 11$~MW), as predicted---and $255$ in which \emph{every} branch margin remains positive, an independent mechanism (Section~\ref{sec:limits}). Every negative entry lies in the active-power block: across all $1594$ evaluated cells the reactive block exhibited none, the empirical face of $\qstar$. The ordered-pair test---finite dominated increments, the form the interval assumption actually requires---separates the two orders: along the cone of \cite{dvijotham2017} \emph{no} violation occurs on either benchmark ($0/12\,000$), while componentwise the IEEE 33-bus feeder passes ($0/7000$) and CIGRE fails ($364/6000$; worst decrease $1.4\times10^{-2}$ in $v$).

\begin{figure}[t]
\centering
\includegraphics[width=0.82\columnwidth]{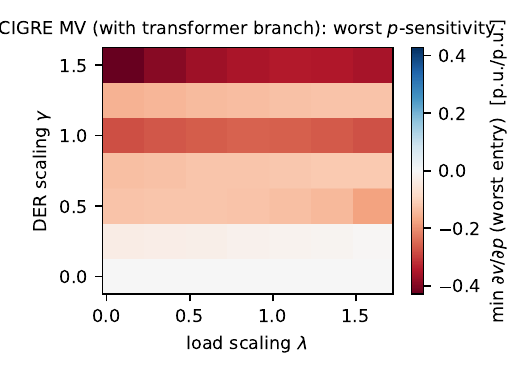}
\caption{Worst active-axis sensitivity per $(\lambda,\gamma)$ cell on the CIGRE MV network (station branch). Blue is monotone ($\partial v/\partial p>0$); red is reversed. The reversal fills the export half-plane ($\gamma$ large) and is essentially independent of load scaling $\lambda$, consistent with an export-driven branch threshold rather than a loadability limit. White cells admit no power-flow solution.}
\label{fig:heatmap}
\end{figure}

\subsection{A second independent network and a conductance sweep}\label{sec:secondnet}
To test whether the reversal is a property of the CIGRE parametrization or of the branch conductance itself, we repeated the analysis on \texttt{mv\_oberrhein}, a real German 20\,kV distribution network (179 buses, two 110/20\,kV, 25\,MVA station transformers) unrelated to the CIGRE family, and on a controlled conductance sweep. Expectations and acceptance criteria were fixed before the runs. The station transformer of \texttt{mv\_oberrhein} has $\kappa = 0.025$---nearly twice the CIGRE value---so Theorem~\ref{thm:sign} predicts a proportionally higher onset, $\pstar = g\,v \approx 5.8$~MW against $2.8$~MW on CIGRE. The measured receiving-end voltage on the station busbar peaks at a transformer export of $5.80$~MW, matching the predicted $5.84$~MW to $0.7\,\%$; the reversal is again confined to the station branch (station sensitivity $-4.4\times10^{-6}$, a distant feeder branch $+7.3\times10^{-3}$). Sweeping the station-branch conductance over $\kappa \in [0.005, 1.0]$, the observed onset tracks the predicted $g\,v$ with a maximum relative deviation of $0.004\,\%$ (Fig.~\ref{fig:kappasweep}). The threshold is thus governed by the closed form $g\,v$, not by a single network's parameters, and it transfers to an independent real network on a different transformer.

\begin{figure}[t]
\centering
\includegraphics[width=0.68\columnwidth]{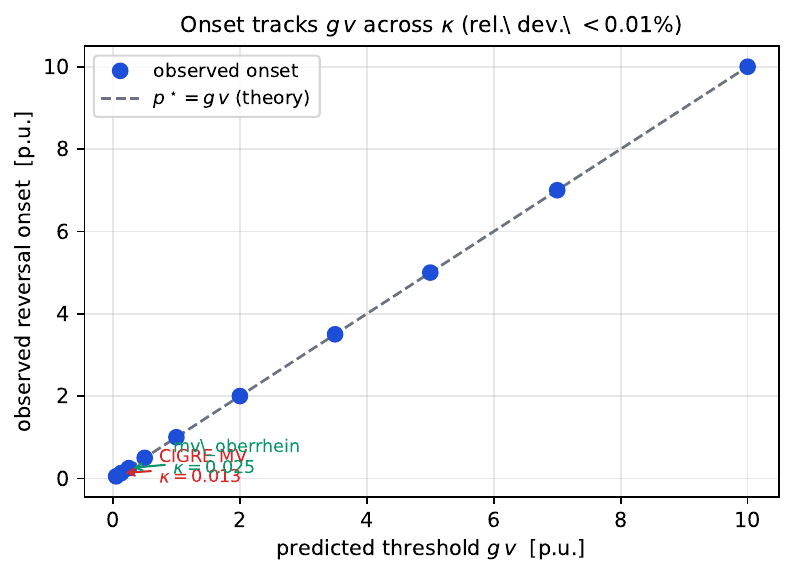}
\caption{Observed reversal onset vs.\ the predicted threshold $g\,v$ as the station-branch conductance is swept ($\kappa \in [0.005,1.0]$). Points lie on the identity line to within $0.01\,\%$. The two labelled points are the two independent real networks: CIGRE MV ($\kappa=0.013$) and \texttt{mv\_oberrhein} ($\kappa=0.025$).}
\label{fig:kappasweep}
\end{figure}

\section{Restoration Behind the Regulated Busbar}\label{sec:restoration}
Corollary~\ref{cor:thresh} suggests a delimitation: if the low-$\kappa$ branch is removed by treating the regulated MV busbar as a fixed-voltage boundary (as an on-load tap changer effectively enforces), the remaining branches carry unreachable thresholds and monotonicity should be restored. We extracted both CIGRE feeders behind their busbars and re-ran the \emph{full} protocol on each. Feeder~1 ($10$ buses, $\kappa = 0.70$) passes $770/770$ grid cells and $152/152$ last-decade cells up to $0.99\,\lsn$, worst sensitivity $+0.048$; feeder~2 ($2$ buses, $\kappa = 1.39$) likewise passes $770/770$ and $152/152$, worst $+0.089$. Neither produces a single non-converged cell, against $214$ on the full network: non-existence is itself an artifact of the station branch. The same network with the transformer branch in the tree violates $p$-monotonicity in $70\,\%$ of solvable cells (Fig.~\ref{fig:kappa}). Every violating cell involves the station branch, supporting the operational formulation: \emph{componentwise monotonicity holds per feeder behind a regulated busbar; across the station branch, every threshold-attributed violation in this protocol was preceded by the flow exceeding $g\,v$, but---as Section~\ref{sec:limits} shows---staying below $g\,v$ does not suffice.}

\begin{figure}[t]
\centering
\includegraphics[width=0.62\columnwidth]{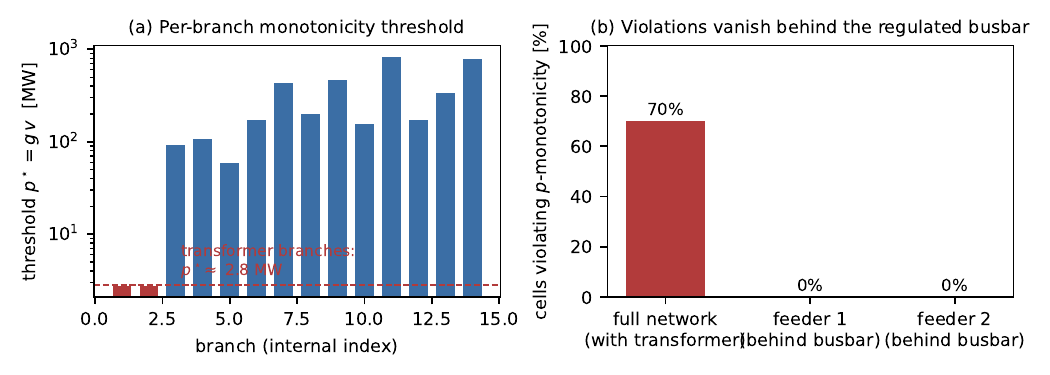}
\caption{(a) Per-branch monotonicity threshold $\pstar = g\,v$ of the CIGRE MV network (log scale); the transformer branches (red) are the only ones with reachable thresholds. (b) Fraction of \emph{solvable} operating cells violating $p$-monotonicity: full network ($390/556$; the remaining $214$ of $770$ admit no solution) vs.\ the two feeders extracted behind the regulated busbar ($0/770$ each, all cells solvable).}
\label{fig:kappa}
\end{figure}

\section{Limits of the Network Extension}\label{sec:limits}
Is the per-branch threshold \emph{sufficient} network-wide? It is not. For a root branch (upstream voltage fixed at the slack), Theorem~\ref{thm:sign} applies verbatim with the nodal injection replaced by the receiving-end flow $(F_1, G_1)$, and the chain rule gives, exactly,
\begin{equation}\label{eq:chain}
\frac{\partial v_1}{\partial p_j} = \frac{2|z_1|^2}{\sqrt{\Delta_1}}\Bigl[(g_1 v_1 - F_1)\frac{\partial F_1}{\partial p_j} + (b_{a,1} v_1 - G_1)\frac{\partial G_1}{\partial p_j}\Bigr].
\end{equation}
Both margins can be positive, far below their thresholds; the sign then rests on the flow derivatives. The active term decays as losses absorb a growing share of the injection increment; the reactive term is \emph{negative}, since reactive losses $x\,\ell$ grow with the squared current and pull the root-directed reactive flow down. Its weight $b_{a}v - G$ is of short-circuit scale---the asymmetry that makes $\qstar$ unreachable---so a lossy station branch amplifies it ($b_{a}/g = 1/\kappa \approx 75$ here). On the IEEE 33-bus feeder a single-node export at bus~18 flips sensitivities at $11.4$~MW with every branch flow below $15\,\%$ of its threshold---but at $|V| \approx 1.38$~p.u., outside any operational envelope. On the CIGRE MV network it operates \emph{inside} the protocol: the $255$ margin-positive cells of Section~\ref{sec:evidence} are its instances, $75$ with \emph{all} node voltages in $[0.90,1.10]$~p.u. Reconstructing \eqref{eq:chain} on all $390$ violating cells, for both station branches, matches the finite differences to $3\times10^{-8}$; in all $255$ the reactive term dominates and the negative entries sit at the station node. No violation was ever localized strictly inside a feeder. We therefore state the network extension as an explicit hypothesis, not a theorem: \emph{per feeder behind a regulated busbar, for operating points within an operational voltage envelope, positive margins $g_e v_e - f_e > 0$ on every branch imply componentwise monotonicity in $p$.} All feeder-level evidence is consistent with this form; the IEEE 33-bus counterexample shows the envelope qualifier cannot be dropped. Quantifying it from \eqref{eq:chain} is ongoing work.

\section{Operational Reading and Implications}\label{sec:implications}
The threshold $\pstar = g\,v$ has a direct operational reading. For any branch $e$ carrying active flow $P_e$, the dimensionless ratio
\begin{equation}\label{eq:margin}
\eta_e := \frac{P_e}{\pstar_e} = \frac{P_e}{g_e v_e}
\end{equation}
is a \emph{monotonicity margin}: while $\eta_e < 1$ on every branch of a feeder behind a regulated busbar, extreme-point screening over that feeder is valid; once $\eta_e \ge 1$ on the station branch, the corners no longer bound the interior and full nonlinear evaluation is required. Because $\pstar_e = g_e v_e$ is available in closed form from data an operator already holds---branch $r,x$ and the local voltage---$\eta_e$ can be evaluated per branch without a sweep, and flags exactly the low-$\kappa$ branches (station transformers, short thick cables) on which the reversal is reachable. A distribution-management reading of \eqref{eq:margin}---as an alarm on incremental hosting studies, and as a map of at-risk transformers---is a natural next step and is developed separately; here we note only that the flag is closed-form and cheap, in contrast to the nonlinear sweep it replaces where $\eta_e<1$. A concrete instance on \texttt{mv\_oberrhein} makes the failure explicit: over the busbar-injection interval $[18,28]$~MW, both endpoints give an identical station-busbar magnitude ($|V|=1.0195$~p.u.), so a corner-only screen reports a flat, monotone interval; yet the interior peaks at $|V|=1.0198$~p.u.\ where the station-branch flow reaches $\pstar$ ($\eta_e \approx 1$), exceeding \emph{both} corners. The excess is small in magnitude on this well-regulated network, but its sign is exactly what a monotone corner screen assumes away, and the margin $\eta_e$ flags the interval without evaluating its interior.

The mechanism is, by construction, invisible to every model affine in the injections. Such models retain the linear voltage drop but discard the quadratic, current-squared loss term---precisely the term through which the reversal runs---and their sensitivity matrices are constant and nonnegative. Screening built on them therefore certifies operating points the nonlinear system does not satisfy, precisely in the low-$\kappa$ regime where $\pstar$ is reachable and, on the CIGRE transformer branch, at nominal voltage, where no voltage limit is active to raise an alarm.

Nor is the failure removed by inverter Volt-VAR control. With the IEEE~1547 Cat.~B default $Q(V)$ characteristic on the CIGRE MV network, the loss-mediated onset is unchanged ($11$--$12$~MW of installed DER) and, at deep export, the reversal is \emph{amplified} by a factor of about $2.5$: the upper droop makes the inverters absorb reactive power into the very channel that drives it. Volt-VAR moreover introduces a \emph{second}, control-mediated violation at low penetration ($3.5$--$4.5$~MW). Its sign follows from the total derivative \eqref{eq:total} with the droop slope $\alpha = \mathrm{d}q/\mathrm{d}p < 0$ active: while an inverter sits on the lower droop, rising $p$ lifts it toward the deadband, so its reactive support $q$ falls ($\alpha<0$); on the reactive-supporting branch $\partial v_1/\partial q>0$, and once $|\alpha\,\partial v_1/\partial q|$ exceeds the still-positive $\partial v_1/\partial p$, the total derivative $\mathrm{d}v_1/\mathrm{d}p$ turns negative even though the componentwise margin has not. This is a controller-trajectory effect, distinct from the Theorem~\ref{thm:sign} threshold; a full treatment of the $Q(V)$ dynamics is left to future work. The present result complements the cone-order monotonicity of \cite{dvijotham2017}, the ray monotonicity of \cite{miu2000}, and the sufficient conditions of \cite{wcg2017}, and sharpens the measured and P--V accounts \cite{iioka2019,matsumura2018,das2022} into an exact componentwise threshold with its $p/q$ asymmetry. Proving the envelope hypothesis is part of a broader program on certified operating margins \cite{mester2026a,mester2026b}.

\section{Conclusion}\label{sec:conclusion}
Componentwise voltage monotonicity in active power is not a property of radial distribution networks; it fails, exactly, at $\pstar = g\,v$ on the receiving end of a branch, and reachably so on the low-$R/X$ station transformer at nominal voltage. The reactive threshold sits a factor $1/\kappa$ higher and is never reached. On two independent real networks the reversal appears at DER penetrations already relevant to hosting studies, is confined to the station branch, tracks the closed-form threshold $g\,v$ across a conductance sweep to within $0.01\,\%$, and is compounded rather than cured by default Volt-VAR control. Because the models most used for fast screening are affine in the injections, they cannot see it. The closed-form margin $\eta_e = P_e/(g_e v_e)$ turns the theorem into a per-branch flag for exactly where extreme-point screening ceases to certify the interior.

\end{document}